\documentclass[conference,a4paper]{IEEEtran}
\usepackage[utf8]{inputenc}
\usepackage[T1]{fontenc}
\usepackage{url}              
\usepackage{cite}             

\usepackage[cmex10]{amsmath}  
\allowdisplaybreaks

\usepackage{mleftright}       
\mleftright                   

\usepackage{graphicx}         
\usepackage{booktabs}         

\usepackage{amssymb}

\usepackage{tikz}
\usepackage{pgfplots}
\usepgfplotslibrary{groupplots,dateplot}

\usepackage{flushend}
\usepackage{bm}
\usepackage{graphicx}         
\usepackage{booktabs}         

\usepackage{caption,subcaption}

\usetikzlibrary{shapes,arrows}
\usetikzlibrary{positioning,fit,calc}
\tikzstyle{block} = [draw, fill=white, rectangle,
minimum height=3em, minimum width=6em]
\tikzstyle{input} = [coordinate]
\tikzstyle{output} = [coordinate]

\newcommand{\E}{ \mathbb{E} } 

\let\leq\leqslant
\let\geq\geqslant

\let\phi\varphi

\newcommand{\card}{\ensuremath M}

\usepackage{comment}

\newtheorem{theorem}{Theorem}
\newtheorem{lemma}{Lemma}

\newtheorem{corollary}{Corollary}

\newtheorem{remark}{Remark}
\newtheorem{example}{Example}

\usepackage{todonotes}

\begin{document}

\title{\textit{What can Information Guess? }Guessing Advantage vs. Rényi Entropy for Small Leakages}

 \author{
	   \IEEEauthorblockN{Julien Béguinot and Olivier Rioul
	   \IEEEauthorblockA{LTCI, Télécom Paris, Institut Polytechnique de Paris, France
	   					\textit{firstname.lastname}@telecom-paris.fr}}\\[-3.65ex]
	 }
\maketitle

\begin{abstract} 
%
%
We leverage the Gibbs inequality and its natural generalization to Rényi entropies to derive closed-form parametric expressions of the optimal lower bounds of $\rho$th-order guessing entropy (guessing moment)
of a secret taking values on a finite set,  in terms of the Rényi-Arimoto $\alpha$-entropy.
This is carried out in an non-asymptotic regime when side information may be available.
The resulting bounds yield a theoretical solution to a fundamental problem in side-channel analysis: Ensure that an adversary will not gain much guessing advantage when the leakage information is sufficiently weakened by proper countermeasures in a given cryptographic implementation.
Practical evaluation for classical leakage models show that the proposed bounds greatly improve previous ones for analyzing  the capability of an adversary to perform side-channel attacks.
\end{abstract}

\maketitle

\section{Introduction}

\emph{Guessing entropy}~\cite{massey1994guessing}, also known as \emph{guesswork}~\cite{pliam2000gueswork}, is perhaps the most popular security metric 
in the context of side-channel attacks of embedded cryptographic devices, such as the cryptographic microcontrollers used in banking smartcards~\cite{standaert+2009unified,ChoudaryPopescu17,TanasescuChoudaryRioulPopescu21}.
Such attacks exploit \emph{leakage} information to recover the secret key, byte by byte in a divide-and-conquer strategy, in which each subkey byte $K\in\{1,\ldots,M\}$ 
(typically $M=128$ or $256$) 
is targeted independently of the others. The secret key $K$ is generally  assumed uniformly distributed, but in a more general framework, any type of secret $X\in\{1,\ldots,M\}$ (e.g., passwords, sensitive personal information, etc.) can be targeted from some disclosed side information $Y$.

Guessing entropy~\cite{massey1994guessing} $G(X|Y)$, or more generally, $\rho$th order guessing moments~\cite{arikan1996} $G_\rho(X|Y)$, relate to the number of tries that the attacker has to make to find the actual secret~$X$ for a given leakage side information $Y$, thereby estimating the brute force effort to find $X$ by exhaustive search. The popularity of $G_\rho$ (particularly $G=G_1$) as a security criterion comes from the fact that it is particularly informative, as it is computed from whole key ranking distribution~\cite{poussier2016keyenum}.
The adversary's \emph{guessing advantage} is defined as 
\begin{equation}
\Delta G_\rho (X;Y) \triangleq G_\rho(M)-G(X|Y) 
\end{equation}
where $G_\rho(M)=G_\rho(K)$ is for blind estimation (no leakage) of a uniformly distributed $M$-ary secret. 

Side channel leakage was evaluated by information theoretic measures such as mutual information $I(X;Y)$~\cite{standaert+2009unified,eloi2019info,eloi2019ches,wei2022mask}, maximal leakage%
~\cite{issawagnerkamath2020operational,julien2023maxleak}, 
and more generally, Arimoto's $\alpha$-information~\cite{arimoto1977information} $I_\alpha(X;Y)=H_\alpha(X)-H_\alpha(X|Y)$ or Sibson's $\alpha$-information~\cite{Sibson69,Verdu15,sankar2019alpha,yi2021conditionalalpha}.
These two $\alpha$-informations coincide for uniform secrets: $I_\alpha(K;Y)=\log M -H_\alpha(K|Y)$.
In many practical cases, such as for protected implementations with masking or low SNR noise~\cite{ishai2003private,wei2022mask,Masure2023nearlytight,Masure2023removing,julien2023maxleak}, 
the information leakage is \emph{small}, which means that the conditional Rényi-Arimoto entropy~\cite{fehrberens2015arimoto} $H_\alpha(X|Y)$ approaches its maximum value  $\log M$ in the case of uniform secrets. Thus, to evaluate the impact of leakage in cryptographic implementations, a significant quantity  is the ``\emph{information advantage}''
\begin{equation}
 \Delta H_\alpha(X;Y) \triangleq \log M - H_\alpha(X|Y).
\end{equation}

A fundamental problem in side-channel analysis is to ensure that an adversary will not gain much guessing advantage when the leakage is sufficiently weakened by proper countermeasures in some cryptographic implementation~\cite{ChoudaryPopescu17,TanasescuChoudaryRioulPopescu21}.
In other words, it is important to upper bound $\Delta G_\rho (X;Y)$ in terms of $\Delta H_\alpha(X;Y)$, or equivalently, to lower bound $G_\rho (X|Y)$ in terms of $H_\alpha(X|Y)$.
As explained in~\cite{ChoudaryPopescu17,TanasescuChoudaryRioulPopescu21}, this is also useful to practically evaluate guessing advantage for a full key (e.g., 256-bit key) whose direct evaluation is not tractable.

Many such bounds have been derived in the literature. Massey's original inequality~\cite{massey1994guessing} is for $\alpha=\rho=1$ and was improved by Rioul~\cite{rioul2022variations} as an asymptotically optimal inequality~\cite{TanasescuChoudaryRioulPopescu21} as $M\to+\infty$.
Arikan's inequalities~\cite{arikan1996} are for $\alpha=\frac{1}{1+\rho}$ and exhibit asymptotic equivalence. More general bounds for various ranges of $\alpha$ and $\rho$ were established in~\cite{rioul2022variations}. These works, however, can only be optimal as $M\to+\infty$ and do not consider the nonasymptotic scenario for small leakage, i.e., around the corner $(H_\alpha(X|Y)\approx \log M,G_\rho(X|Y)\approx G_\rho(M))$ for relatively small $M$.

\begin{figure}[ht!]
    \centering
    \includegraphics[width=.4\textwidth,height=.25\textwidth]{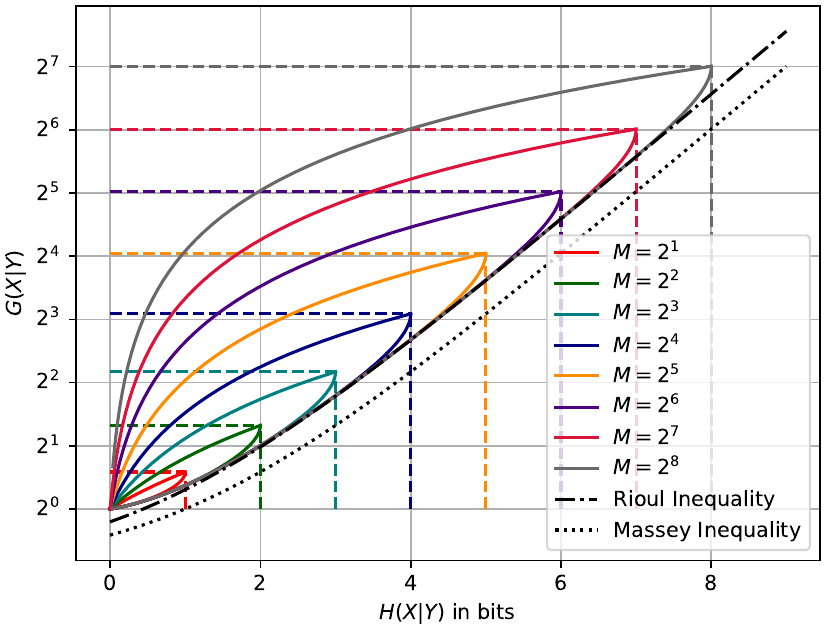}
    \caption{\small Optimal joint range region between $H(X|Y)$ and $G(X|Y)$ for different values of $M$. 
The (optimal) \emph{upper} bound is that of Mc Eliece and Yu~\cite{mceliece1995inequality}. The optimal lower bound is derived in this paper.
The black dotted and dash-dotted curves correspond to Massey's~\cite{massey1994guessing} and Rioul's~\cite{rioul2022variations} inequalities, that do not depend on $M$.  }
    \label{fig:motivation}
\end{figure}

As an illustration for $\alpha=\rho=1$, Figure~\ref{fig:motivation} shows that for any fixed value of $M$, there is a multiplicative gap of approximately $e/2$ in guessing entropy from this corner point $(\log M, \frac{M+1}{2})$ compared to Rioul's inequality.
Sason and Verd{\'u}~\cite{sason2018improved} improved Arikan's inequalities in a non-asymptotic regime
for any ranking function but did not obtain the exact locus of attainable values of $G_\rho$ vs. $H_\alpha$, nor closed-form expressions for lower bounds of $G_\rho$ vs. $H_\alpha$.

In this paper, we leverage the Gibbs inequality~\cite{CoverThomas} and its natural generalization to Rényi entropies~\cite{rioul2020renyiEPI} to derive closed-form parametric expressions of the optimal lower bounds of $G_\rho(X|Y)$ vs. $H_\alpha(X|Y)$ (or upper bounds of $\Delta G_\rho(X ;Y)$ vs. $\Delta H_\alpha(X ;Y)$).
We obtain an explicit, easily computable, first-order bound of the form $\Delta G_\rho(X;Y)\leq c \sqrt{\Delta H_\alpha(X;Y)}$ for small leakages. We then evaluate and refine our bounds in the 
Hamming weight and random probing leakage models.


\section{Notations}

The secret is modeled as a random variable $X$  with pmf $p_X$ taking values in a finite set $\{1,\ldots,M\}$. The side channel $X\to Y$ leaks some side information $Y$, which is modeled as an arbitrary random variable (discrete or continuous).
For most cryptographic applications, $X$ is a uniformly distributed key $K \sim \mathcal{U}(M)$.
In this case the side channel is noted $K \to Y$.

We let $\alpha > 0$ and $\rho > 0$ be the entropy and guessing entropy orders, respectively. 
The Rényi entropy of order $\alpha$ is $H_\alpha(X) \triangleq - \alpha' \log \| p_X \|_\alpha$ where
$ \| p \|_\alpha \triangleq  ( \sum p^\alpha  )^{{1}/{\alpha}}$ denotes the ``$\alpha$-norm'', the logarithm is taken to an arbitrary base, and $\alpha'$ is the Hölder conjugate of $\alpha$: $\frac{1}{\alpha} + \frac{1}{\alpha'} = 1$, i.e. $\alpha' = \frac{\alpha}{\alpha-1}$ or $(\alpha'-1)(\alpha-1) = 1$. We also write $H_\alpha(p)=- \alpha' \log \| p \|_\alpha$. The limiting case $\alpha\to 1$ is the Shannon entropy $H(X)=H_1(X)$. 
The limiting case $\alpha\to\infty$ is the min-entropy $H_\infty(X) \triangleq  - \log \max p_X$.

The $\rho$-guessing entropy  (a.k.a. guessing $\rho$th-order moment) is 
$G_\rho(X) \triangleq \min_{ \sigma} \sum_{i = 1}^{\card} p_X( \sigma(i) ) i^\rho =  \min_{ \sigma} \E(\sigma(X)^\rho)$, where the minimum is over all possible permutations of the secret values.
We also write $G_\rho(p)$ when $X$ has pmf $p$. By the rearrangement inequality,
since the sequence $i^\rho$ is increasing, the optimal guessing strategy is such that the $p_X( \sigma(i) )$ are arranged in decreasing order. 
 When $X=K$ is uniformly distributed, its $\rho$-guessing entropy is also noted
$G_\rho(M) \triangleq \frac{1}{M} \sum_{i = 1}^{M} i^\rho$.

In the presence of leakage $Y$, the most natural definition of conditional $\alpha$-entropy is that of Arimoto~\cite{arimoto1977information,fehrberens2015arimoto}: 
$H_\alpha(X|Y) \triangleq - \alpha'  \log \E_Y  \| p_{X|Y} \|_\alpha$.
In particular 
 $H_\infty(X|Y) \triangleq  - \log \E_Y \max p_{X|Y}$ where $\E_Y \max p_{X|Y}$ is the optimal success probability as given by the MAP rule.
Following Hirche~\cite{hirche20renyi} we also define
\begin{align}\label{eqKalpha}
    K_\alpha(X) &\triangleq \|p_X\|_\alpha =\exp \bigl ( \tfrac{1-\alpha}{\alpha} H_\alpha(X) \bigr )  \\
    K_\alpha(X|Y) &\triangleq\E_Y  \|p_{X|Y}\|_\alpha = \exp \bigl ( \tfrac{1-\alpha}{\alpha} H_\alpha(X|Y) \bigr )  
\end{align}
so that $K_\alpha(X|Y) = \E_y K_\alpha(X|Y=y)$.
Likewise, the conditional $\rho$-guessing entropy is 
$G_\rho(X|Y) \triangleq \E_{y }  G_\rho(X|Y=y) $.


\section{Optimal Lower Bound on $G_\rho(X|Y)$ vs. $H_\alpha(X|Y)$}

\subsection{$G(X|Y)$ vs. $H(X|Y)$}

\begin{theorem}\label{th:GH}
The optimal lower bound on $G(X|Y)$ vs. $H(X|Y)$ is given by the parametric curve for $\gamma\in(0,1)$:
\pagebreak
\begin{equation}\label{eq:GH}
    \begin{cases}
        G(X|Y) = \frac{1}{1-\gamma} - \frac{M \gamma^M}{1-\gamma^M} 
        \\
        H(X|Y) = \log (\! \gamma \frac{1 - \gamma^{\card}}{1 -\gamma} ) \!-\! (\log \gamma) ( \frac{1}{1-\gamma} \!-\! \frac{M \gamma^M}{1-\gamma^M} )
    \end{cases} 
\end{equation}
where the limiting case $\gamma\to 1$ gives $G=\frac{M+1}{2}$ and $H=\log M$ attained for the uniform distribution. 

The optimal upper bound on $\Delta G(X ; Y) = \frac{M+1}{2}-G(X|Y)$ vs. $\Delta H(X ; Y) = \log M -H(X|Y)$ is given by the parametric curve for $\mu\in(0,+\infty)$:
\begin{equation}\label{eq:DeltaGDeltaH}
\!\!\begin{cases}
    \Delta G(X ; Y) =  \frac{1}{2} \bigl( \card \coth( \card \mu)  -  \coth( \mu ) \bigr)
\\
    \Delta H(X;Y) = \log  \frac{ \card \sinh \mu}{\sinh(\card \mu)}  
 +2\mu (\log e) \Delta G(X;Y).
\end{cases}
\end{equation}
\end{theorem}

\begin{proof}
Here $\alpha=\rho=1$ is the classical situation studied by Massey~\cite{massey1994guessing}, where
the essential ingredient is the \emph{Gibbs inequality}~\cite{CoverThomas}:
 In the unconditional case, for any pmf $q$, 
 \begin{equation} 
 H( X ) \leq - \E_X \log q(X)
 \end{equation} 
 with equality iff $q=p_X$.
One may always assume that $p_X(x)$ is nonincreasing in $x$, in which case $G(X)=\E(X)$.
Therefore, we choose $q$ such that $\log q(x) = a + b x$ for real constants $a,b$, so that $\E_X \log q(X) = a + b G(X)$. To allow equality in the Gibbs inequality ($q=p_X$), it is necessary that $q(x)$ be nonincreasing in $x$, i.e., $b\leq 0$. Now 
$q(x)$ rewrites as a truncated geometric pmf $q(x)= c_\gamma \gamma^x$ where $\gamma = \exp(b) \in (0,1]$ and $c_\gamma = ( \sum_{i=1}^M \gamma^i)^{-1} > 0$
is a normalization factor. 
Thus one obtains 
\begin{equation}
    H(X) \leq -\log c_\gamma - (\log \gamma) G(X).
\end{equation}
for all $\gamma \in (0,1)$, with equality iff $p_X=q$ where $H(X)=H(q)$ and $G(X)=G(q)$.

In the conditional case, we similarly have
$H(X|Y=y) \leq -\log c_\gamma - (\log \gamma) G(X|Y=y)$
 for every $y$. Taking the expectation over $Y$ yields the same inequality for conditioned entropies:
\begin{equation}
H(X|Y) \leq -\log c_\gamma - (\log \gamma) G(X|Y),
\end{equation}e.
that is, $G(X|Y) \geq - (\log \gamma)^{-1} ( H(X|Y) + \log c_\gamma)$.
The equality case still corresponds to $H(X|Y)=H(q)$ and $G(X|Y)=G(q)$.
Since $q$ approaches the uniform distribution as $\gamma \to 1$ and the Dirac distribution as $\gamma \to 0^+$, all possible values of entropies are attainable.

We thus obtain the following parameterization of the optimal lower bound on $G(X|Y)$ vs. $H(X|Y)$:
\begin{equation}
\begin{cases}
 G(X|Y) = G(q)= c_\gamma \sum_{i=1}^M  i \gamma^i 
 \\
 H(X|Y) = H(q)= - \log c_\gamma - G(q) \log \gamma
\end{cases}
\end{equation}
where $\gamma \in (0,1]$. The case $\gamma=1$ gives $H=\log M$ and $G=\frac{M+1}{2}$ attained for the uniform distribution. For $0<\gamma<1$, a straightforward calculation gives~\eqref{eq:GH}.
Setting $\mu \triangleq - \frac{1}{2} \ln \gamma \in (0,\infty)$, i.e., $\gamma=e^{-2\mu}$ gives~\eqref{eq:DeltaGDeltaH} for $\Delta G(X ; Y) = \frac{M+1}{2}-G(X|Y)$ and $\Delta H(X ; Y) = \log M -H(X|Y)$.
\end{proof}

\subsection{$G_\rho(X|Y)$ vs. $H(X|Y)$}

\begin{theorem}\label{th:GrhoH}
    The optimal lower bound of $G_\rho(X|Y)$ vs. $H(X|Y)$ is given by the parametric curve for $\gamma \in (0,1]$:
\pagebreak
    \begin{equation}\label{eq:GrhoH}
        \begin{cases}
            G_\rho(X|Y) = (\sum_{i=1}^M  i^{\rho} \gamma^{i^\rho}) (\sum_{i=1}^M \gamma^{i^\rho})^{-1}  
            \\
            H(X|Y) = \log ( \sum_{i=1}^M \gamma^{i^\rho} ) - (\log \gamma) \frac{\sum_{i=1}^M i^{\rho} \gamma^{i^\rho}}{ \sum_{i=1}^M \gamma^{i^\rho} }
        \end{cases}
    \end{equation}
\end{theorem}

\begin{proof}
The proof is analog to the case of the preceding subsection, where  $\log q(x) = a + b x^\rho$, $q(x)=c_\gamma \gamma^{x^\rho}$, $\gamma = \exp b \in (0,1]$ and $c_\gamma =  ( \sum_{i=1}^M \gamma^{i^\rho} )^{-1} > 0$. Again $q$ approaches the uniform distribution as $\gamma \to 1$ and the Dirac distribution as $\gamma \to 0^+$, hence all possible values of entropies are attainable. One readily obtains
\begin{equation}
\begin{cases}
 G_\rho(X|Y) = G_\rho(q)= c_\gamma \sum_{i=1}^M  i^{\rho} \gamma^{i^\rho} 
 \\
 H(X|Y) = H(q)= - \log c_\gamma -  G_\rho(q) \log \gamma
\end{cases}
\end{equation}
which gives~\eqref{eq:GrhoH}.
\end{proof}

\subsection{$G_\rho(X|Y)$ vs. $H_\alpha(X|Y)$}

The most general optimal lower bound is given by the following Theorem.
It generalizes Theorems~\ref{th:GH} and~\ref{th:GrhoH}, which can be recovered in the limiting case $\alpha\to 1$.
\begin{theorem}\label{th:GrhoHalpha}
When $0<\alpha<1$, the optimal lower bound of $G_\rho(X|Y)$ vs. $H_\alpha(X|Y)$ is given by the parametric curve for $\gamma \in (0,\infty)$:
\begin{equation}\label{eq:GrhoHalpha<1}
    \begin{cases}
        G_\rho(X|Y) = 1+ \gamma^{-1} \bigr ( \frac{ \sum_{i=1}^M  ( 1-\gamma + \gamma i^\rho  )^{\alpha'}  }{  \sum_{i=1}^M  ( 1-\gamma + \gamma i^\rho  )^{\alpha'-1} } - 1 \bigr ) 
        \\
        H_\alpha(X|Y) = \alpha' 
        \log \sum_{i=1}^M  ( 1 -\gamma + \gamma i^\rho  )^{\alpha'-1}
        \\
        \hphantom{H_\alpha(X|Y) } 
        + (1-\alpha') \log \sum_{i=1}^M  ( 1 -\gamma + \gamma i^\rho  )^{\alpha'}.
    \end{cases}
\end{equation}
When $\alpha > 1$, the optimal lower bound of $G_\rho(X|Y)$ in terms of $H_\alpha(X|Y)$ is given by the parametric curve for $\gamma \in (0,1)$:
\begin{equation}\label{eq:GrhoHalpha>1}
    \begin{cases}
        G_\rho(X|Y) = \gamma^{-1} \bigl ( 1- \frac{ \sum_{i=1}^M (1-\gamma i^\rho)_+^{\alpha'}  }{  \sum_{i=1}^M (1-\gamma i^\rho)_+^{\alpha'-1} }  \bigr )
        \\
        H_\alpha(X|Y) = \alpha' \log \sum_{i=1}^M (1-\gamma i^\rho)_+^{\alpha'-1}
        \\
        \hphantom{H_\alpha(X|Y) } 
         + (1-\alpha') \log \sum_{i=1}^M (1-\gamma i^\rho)_+^{\alpha'}
    \end{cases}
\end{equation}
where $x_+=\max(x,0)$ denotes the positive part of $x$.
\end{theorem}

The proof relies on the following Gibbs inequality for Rényi entropies\cite[Prop.~8]{rioul2020renyiEPI}:
\begin{lemma}[Generalized Gibbs Inequality]
For any pmf $q$,
\begin{equation}\label{eqn:alpha-Gibbs}
H_\alpha(X) \leq - \alpha' \log \E_{X} q_\alpha^{{1}/{\alpha'}}(X) 
\end{equation}
with equality iff $p_X=q$. Here $q_\alpha$ is the \emph{escort distribution}~\cite{rioul2020renyiEPI} of $q$, defined by 
$q_\alpha(x)  =q^\alpha(x) /   \| q \|_\alpha^{\alpha}$.
\end{lemma}
The proof given~\cite{rioul2020renyiEPI} was for random variables having pdfs with respect to the Lebesgue measure, but applies verbatim to discrete random variables having pmfs with respect to the counting measure on $\{1,2,\ldots,M\}$.
The Gibbs inequality~\eqref{eqn:alpha-Gibbs} can be easily rewritten directly in terms of $q(x)$ as
\begin{equation}\label{eq:GibbsHalpha}
H_\alpha(X) \leq (1-\alpha) H_\alpha(q) - \alpha' \log \E_X q^{\alpha-1}(X).
\end{equation}
In terms of~\eqref{eqKalpha} it also rewrites
\begin{equation}\label{eq:GibbsKalpha}
K_\alpha(X)  \lessgtr  \E_{X} q_\alpha^{{1}/{\alpha'}}(X)  =\frac{ \E_X q^{\alpha-1}(X)}{\|q\|_\alpha^{\alpha-1}}
\end{equation}
where $\lessgtr$ denotes $\geq$ for $0<\alpha<1$ and $\leq$ for $\alpha>1$.

\begin{proof}[Proof of Theorem~\ref{th:GrhoHalpha}]
First consider  the unconditional case. 
One may always assume that $p_X(x)$ is nonincreasing in $x$, in which case $G_\rho(X)=\E(X^\rho)$.
Therefore, we wish to choose $q$ such that $q^{\alpha-1}(x) = a + b x^\rho$ for real constants $a,b$, i.e., $q(x)=(a + b x^\rho)^{\alpha'-1}$, so that $\E_X q^{\alpha-1}(X)=a + b\,G_\rho(X)$.

When $0<\alpha<1$, $\alpha'<0$, to allow equality in the Gibbs inequality ($p_X=q$), it is necessary that $a + b x^\rho\geq 0$ with nonempty support and that $(a + b x^\rho)^{\alpha'-1}$ is nonincreasing for $x=1,2,\ldots,M$.
This gives the conditions $b\geq 0$ and $a>-b$. Rewriting $a + b x^\rho = (a+b) + b(x^\rho-1)$ we obtain
$q(x) =c_\gamma (1 + \gamma (x^\rho -1))^{\alpha'-1}$ where $\gamma \triangleq \frac{b}{a+b}\in[0,+\infty)$ and $c_\gamma =  ( \sum_{i=1}^M (1 + \gamma (i^\rho-1))^{\alpha'-1}  )^{-1} > 0$ is a normalization factor.
Note that $q$ is the uniform distribution when $\gamma =0$ and approaches the Dirac distribution as $\gamma \to +\infty$, hence all possible values of entropies are attainable.

When $\alpha>1$, $\alpha'>0$, we choose the positive part $q(x)=(a + b x^\rho)_+^{\alpha'-1}$. 
Again to allow equality in the Gibbs inequality ($p_X=q$), it is necessary $q(x)$ has nonempty support and is nonincreasing for $x=1,2,\ldots,M$. This gives the conditions $b\leq 0$ and $a>-b$. Factoring out $a>0$ we obtain $q(x) =c_\gamma (1 - \gamma x^\rho)_+^{\alpha'-1}$ where  $\gamma\triangleq \frac{-b}{a} \in (0,1)$ and $c_\gamma =  ( \sum_{i=1}^M (1 - \gamma x^\rho)_+^{\alpha'-1} )^{-1} > 0$ is a normalization factor.
Note that $q$ approaches the uniform distribution when $\gamma \to 0$ and the Dirac distribution when $\gamma \to 1$---in fact, it is the Dirac distribution for all $\gamma \in [2^{-\rho},1)$. Hence all possible values of entropies are again attainable.

In both cases,
the Gibbs inequality~\eqref{eq:GibbsKalpha} takes the form\footnote{%
For $\alpha>1$, the Gibbs inequality takes this form because when $p_X=q$, these pmfs have the same support, so that equality also holds in the inequality 
$\E_X  (a + b X^\rho)_+      \geq \E_X (a + b X^\rho ) \iff  \E_X q^{\alpha-1}(X) \geq a+bG_\rho(X)$.
}
$K_\alpha(X) \lessgtr \phi \bigl(G_\rho(X)\bigr)$
for some linear function~$\phi$,
with equality iff $p_X=q$, in which case $H_\alpha(X)=H_\alpha(q)$ and $G_\rho(X)=G_\rho(q)$.
In the conditional case, we similarly have
$K_\alpha(X|Y=y) \lessgtr \phi \bigl(G_\rho(X|Y=y)\bigr)$
 for every $y$, and taking the expectation over $Y$ yields the same inequality for conditioned entropies:
$K_\alpha(X|Y) \lessgtr \phi \bigl(G_\rho(X|Y)\bigr)$.
The equality case still corresponds to $H_\alpha(X|Y)=H_\alpha(q)$ and $G_\rho(X|Y)=G_\rho(q)$.

When $0<\alpha<1$,  $\E_X q^{\alpha-1}(X) = c_\gamma^{\alpha - 1} (1 + \gamma (G_\rho(X)-1))$\linebreak
and the equality case of the Gibbs inequality~\eqref{eq:GibbsHalpha} rewrites $H_\alpha(q)=-\frac{\alpha'}{\alpha} \log\E_X q^{\alpha-1}(X) = 
 -\frac{\alpha'}{\alpha} \log  \bigl(1 -\gamma+ \gamma G_\rho(q) \bigr )-\log c_\gamma$.
Since by definition $H_\alpha(q)=- \alpha' \log \| q \|_\alpha$ we obtain the parametric curve
\begin{equation}
    \begin{cases}
        G_\rho(X|Y) = \gamma^{-1} ( \|q \|_\alpha^\alpha c_\gamma^{1-\alpha} - 1) + 1
        \\
        H_\alpha(X|Y) = - \alpha' \log \| q \|_\alpha.
    \end{cases}
\end{equation}
Substituting $\| q \|_\alpha^\alpha = c_\gamma^\alpha \sum_{i=1}^M (1+\gamma (i^\rho -1))^{\alpha'}$ gives~\eqref{eq:GrhoHalpha<1}.

When $\alpha > 1$,
$\E_X q^{\alpha-1}(X) = c_\gamma^{\alpha - 1} (1 - \gamma G_\rho(X))$,
and the equality case of the Gibbs inequality~\eqref{eq:GibbsHalpha} rewrites $H_\alpha(q)=-\frac{\alpha'}{\alpha} \log\E_X q^{\alpha-1}(X) = -\frac{\alpha'}{\alpha} \log(1 - \gamma G_\rho(q))-\log c_\gamma$.
We similarly obtain the parametric curve
\begin{align}
    \begin{cases}
        G_\rho(X|Y) = \gamma^{-1} (1- \|q \|_\alpha^\alpha c_\gamma^{1-\alpha})
        \\
        H_\alpha(X|Y) = - \alpha' \log \| q \|_\alpha
    \end{cases}
\end{align}
Substituting $\| q \|_\alpha^\alpha = c_\gamma^\alpha \sum_{i=1}^M (1-\gamma i^\rho)_+^{\alpha'} $
gives~\eqref{eq:GrhoHalpha>1}.
\end{proof}

\begin{remark}
Sason and Verd\'u~\cite{sason2018improved} stated an implicit lower bound on guessing moment vs. Rényi entropy in the unconditional case, 
by specifying only the minimizing pmf (equation~(59) in~\cite{sason2018improved}) for a guessing strategy which is not necessarily optimal. 
Their minimizing pmf is not the same as in the above proof when $0<\alpha<1$ because it is does not satisfy the constraint that $p_X(x)$ should be decreasing in $x$. Additionally, it can be checked that it cannot approach the Dirac distribution, hence does not provide the full range of the entropy values.
\end{remark}


\medskip

As an important consequence, an explicit first-order upper bound on $\Delta G_\rho(X;Y)$ can be obtained, which is easy to compute for any adversary observing small leakages.
\begin{corollary}\label{thm:explicit}
As $\Delta H_\alpha(X;Y) \to 0$, up to first order,
\begin{equation}\label{eq:GrhoHalpha1storder}
    \Delta G_\rho(X;Y) \lesssim  
    \sqrt{ \frac{2 (G_{2\rho}(M)-G^2_\rho(M))}{\alpha} }
     \sqrt{ \frac{\Delta H_\alpha(X;Y)}{\log e} }.
\end{equation}
In particular, $ \Delta G(X;Y) \lesssim  
    \sqrt{\frac{M^2-1}{6 \alpha}}
     \sqrt{ \frac{\Delta H_\alpha(X;Y)}{\log e} }$.
\end{corollary}

\begin{proof} 
One has $\Delta G_\rho(q)\to0$ and $\Delta H_\alpha(q) \to0$ when $q$ approaches the uniform distribution, i.e., when $\gamma\to 0$ in~\eqref{eq:GrhoHalpha<1} or~\eqref{eq:GrhoHalpha>1}.
Taylor expansion about $\gamma=0$ in both cases gives
\begin{equation}
    \begin{cases}
    \Delta G_\rho(X;Y) = \gamma |1 - \alpha'| (G_{2\rho}(M)-G^2_\rho(M)) + O(\gamma^2)
    \\
    \frac{\Delta H_\alpha(X;Y)}{\log e} = \tfrac{|\alpha'(1-\alpha')|}{2} (G_{2\rho}(M)\!-\!G_\rho^2(M)) \gamma^2 + O(\gamma^3)
    \end{cases}
\end{equation}
which yields~\eqref{eq:GrhoHalpha1storder}. This is also valid for $\alpha=1$ by taking the limit $\alpha\to 1$.
\end{proof}

Fig.~\ref{fig:illus-main-theorem} illustrates the results of this section. Next we give two examples for which explicit upper and lower bounds can be derived directly.

\begin{figure}[ht!]
    \centering
    \begin{subfigure}{0.4\textwidth}
        \captionsetup{skip=-0.3cm}%
        \includegraphics[width=\textwidth]{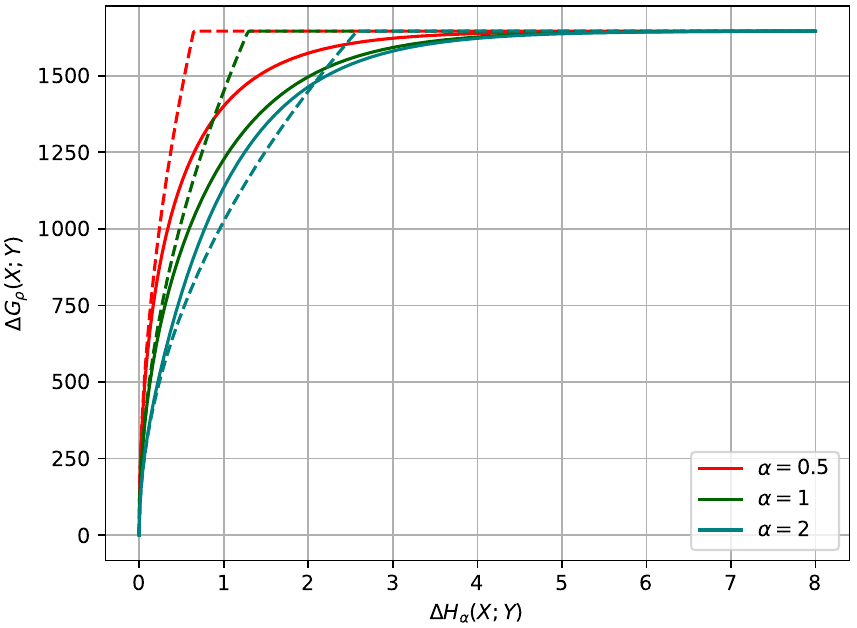}
        \label{fig:illustration_card=8}
        \caption{\small $\rho=1.5$.}
    \end{subfigure}
    \begin{subfigure}{0.39\textwidth}
        \captionsetup{skip=-0.3cm}
        \includegraphics[width=\textwidth]{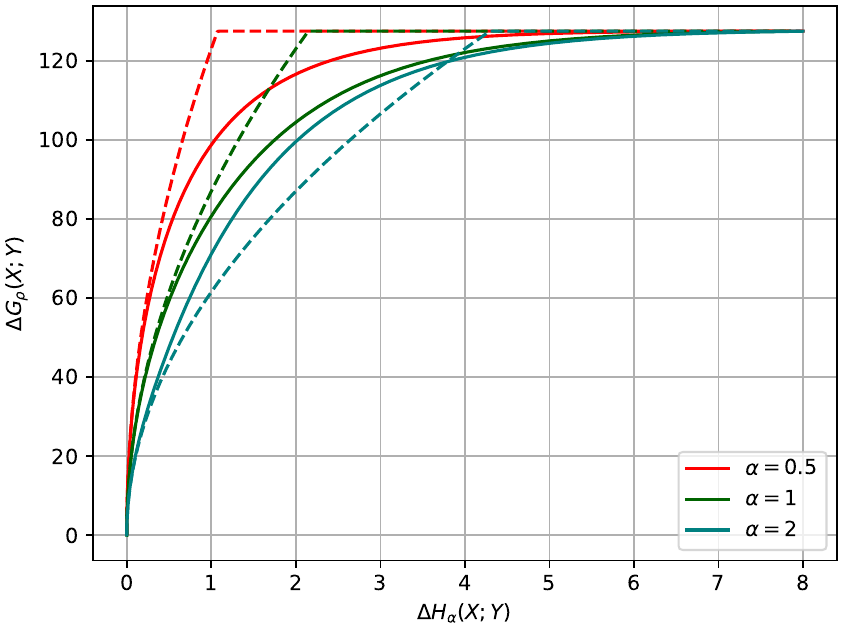}
        \label{fig:illustration_card=8}
        \caption{\small $\rho=1$.}
    \end{subfigure}

    \caption{\small \emph{Solid}: Illustration of Theorems~\ref{th:GH},\ref{th:GrhoH},\ref{th:GrhoHalpha} (upper bounds of $\Delta G_\rho$ vs. $\Delta H_\alpha$) for various values of $\alpha$ and $\rho$ when $M=2^8$. \emph{Dashed}: First-order lower bound $\Delta G_\rho\leq c \sqrt{\Delta H_\alpha}$  from Corollary~\ref{thm:explicit}.}
    \label{fig:illus-main-theorem}
\end{figure}


\begin{example}[Binary Random Variables]
Let $ \mathcal{B}(p)$ be the Bernoulli distribution with parameter $p \in [0,1]$ and define $h_\alpha(p)  \triangleq H_\alpha( \mathcal{B}(p))=\frac{1}{1-\alpha}\log((1-p)^\alpha+p^\alpha)$ and
$k_\alpha(p) \triangleq K_\alpha( \mathcal{B}(p))=((1-p)^\alpha+p^\alpha)^{1/\alpha}$.
We let $h_\alpha^{-1}$ and $k_\alpha^{-1}$ be their inverse when restricted to $[0,\frac{1}{2}]$.
Note that $k_\alpha$ is concave increasing if $\alpha \in (0,1)$; and convex decreasing if $\alpha > 1$.

For $M=2$ in the unconditional case (without side information) we simply have $G_\rho(X)=(1-p)+p2^\rho=1 + (2^\rho -1) p= 1 + (2^\rho -1) k_\alpha^{-1}( K_\alpha(X) )$. 
Since $k_\alpha^{-1}$ is convex, by Jensen's inequality, we obtain the desired lower bound explicitly.
The upper bound is similarly obtained by taking the chord of $k_\alpha^{-1}$:
\begin{equation}
    \begin{aligned}
        1 + (2^\rho -1)& h_\alpha^{-1}( H_\alpha(X|Y) ) 
        \leq G_\rho(X|Y) 
        \\
        &\leq 1 + \tfrac{2^\rho -1}{2} \frac{\exp{ \left ( \frac{1-\alpha}{\alpha} H_\alpha(X|Y) \right ) } -1 }{2^\frac{1-\alpha}{\alpha} - 1}.
    \end{aligned}
\end{equation}
Since $ h_\alpha^{-1}(\log 2 - \delta) \approx \frac{1}{2} - \sqrt{ \frac{1}{2\alpha} \frac{\delta}{\log e } }$ as $\delta \to 0$ we recover Corollary~\ref{thm:explicit} as $\delta=\Delta H_\alpha(X;Y) \to 0$,  which for $M=2$ takes the form
\begin{equation}\label{eqn:taylor2}
    \Delta G_\rho(X;Y) \lesssim 
   (2^\rho -1)     \sqrt{\frac{\Delta H_\alpha(X;Y)}{2\alpha\log e} }.
\end{equation}
In the other direction we also obtain at first order that 
\begin{equation}
\Delta G_\rho(X;Y)\gtrsim
     \tfrac{2^\rho - 1}{2} \tfrac{1-\alpha}{\alpha} \tfrac{2^{\frac{1-\alpha}{\alpha}}}{2^{\frac{1-\alpha}{\alpha}}-1}  \frac{\Delta H_\alpha(X;Y)}{\log e}  
\end{equation}
which reads
$\Delta G_\rho(X;Y)\gtrsim  \tfrac{2^\rho - 1}{2} \frac{\Delta H(X;Y)}{\log e} $
 in the limiting case $\alpha\to 1$.
\end{example}

\begin{example}[Guessing vs. Min-Entropy]
When $\alpha = \infty$ explicit bounds can also be derived.
Sason and Verdù~\cite[Theorem 9]{sason2018improved} provided the joint range between probability of error $\epsilon_{X|Y}$ and $\rho$th-order guessing moments.
Since $H_\infty(X|Y) =  - \log (1-\epsilon_{X|Y})$, their results can be easily rewritten in terms of min-entropy. This gives
\begin{equation}
\begin{aligned}
    K \sum_{i = 1}^{  \lfloor K^{-1} \rfloor } i^\rho &+ ( 1 - K \lfloor K^{-1} \rfloor  ) 
    ( 1 + \lfloor K^{-1} \rfloor )^\rho
    \\
    \leq G_\rho(X|Y) 
    &\leq 1 + \frac{\card}{\card - 1} ( G_\rho(M) -1 )  (1 - K ) 
\end{aligned}
\end{equation}
where $K = K_\infty(X|Y) = \exp \left ( -H_\infty(X|Y) \right )$.
For $\rho =1$, as shown in~\cite{rioul2023interplay}, this simplifies to 
$    
        ( \lfloor  K^{-1} \rfloor + 1   )  ( 1 - \frac{\lfloor  K^{-1} \rfloor}{2} K )
        \leq G(X|Y)
        \leq 1 + \frac{\card}{2} (1 - K).
$

For small leakages 
of uniform secrets with maximal leakage $I_\infty(K;Y) = H_\infty(K) - H_\infty(K|Y) \leq \log \frac{\card}{\card -1} $, this rewrites
$
    \frac{1}{2}  ( e^{I_\infty(K;Y)} - 1  ) \leq  \Delta G(K ; Y) \leq \frac{\card - 1}{2}  ( e^{I_\infty(K;Y)} - 1  ). 
$
As $I_\infty(K;Y) \rightarrow 0$ this yields the first order approximation,
\begin{equation}
    \tfrac{1}{2} \tfrac{I_\infty(K;Y)}{\log e} \lesssim \Delta G(K ; Y) 
   \lesssim \tfrac{\card - 1}{2}  \tfrac{I_\infty(K;Y)}{\log e}.
\end{equation}
The upper bound is linear in $I_\infty(K;Y)$ which agrees with the fact that in~\eqref{eq:GrhoHalpha1storder} the term in $\sqrt{\Delta H_\alpha(K;Y)}=\sqrt{I_\alpha(K;Y)}$ vanishes  as $\alpha\to+\infty$.

\end{example}

\subsection{Evaluation in the Hamming Weight Lekage Model}

A standard leakage model in side-channel analysis~\cite{%
standaert+2009unified,ChoudaryPopescu17,eloi2019info,eloi2019ches,wei2022mask,yi2021conditionalalpha,Masure2023nearlytight,Masure2023removing
}
 is the Hamming weight leakage model $Y = w_H(K) + N$, in which
the Hamming weight $w_H(K)$ of the binary representation of the secret byte $K \sim \mathcal{U}(M)$, $M=2^n$ leaks under additive Gaussian noise $N \sim \mathcal{N}(0,\sigma)$.
Fig.~\ref{fig:HW-model} compares our bounds to previous ones and to the numerical evaluation of the actual guessing advantage.
While the case $\alpha=1$ (Thm.~\ref{th:GH}) is slightly better than $\alpha = \frac{1}{2}$ (Thm.~\ref{th:GrhoHalpha}),  both bounds are very close to the exact value as noise increases, with huge gaps compared to previous Arikan's~\cite{arikan1996} and Rioul's~\cite{rioul2022variations} bounds which were used in~\cite{ChoudaryPopescu17,TanasescuChoudaryRioulPopescu21} for practical evaluation of full-key guessing entropy.

\begin{figure}[ht!]
    \captionsetup{skip=-0.2cm}
    \begin{center}
    \includegraphics*[width=0.4\textwidth,height=.25\textwidth]{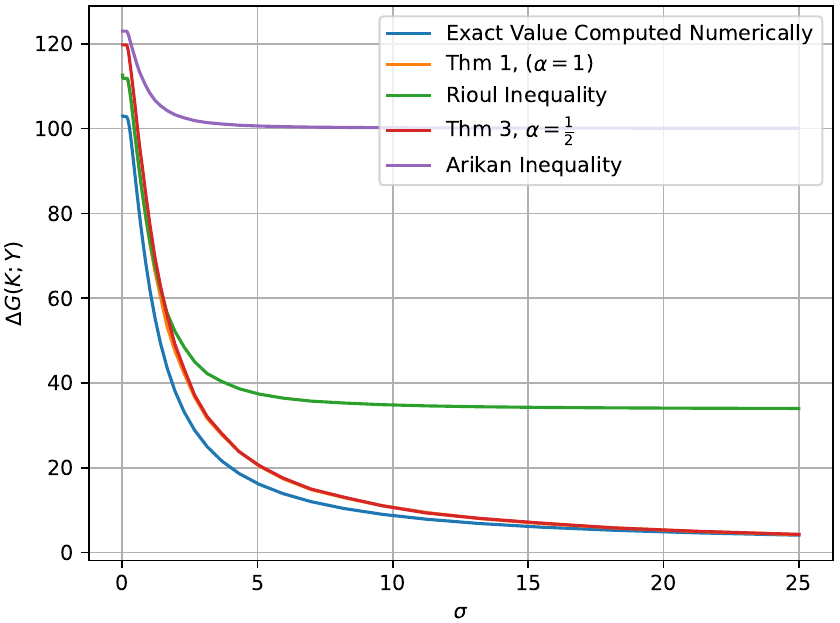}
    \end{center}
    \caption{Bound on the guessing advantage in the Hamming weight leakage model for increasing noise variance $\sigma^2$.}
    \label{fig:HW-model}
\end{figure}


\subsection{Evaluation in the Random Probing Model} 

The random probing model~\cite{DBLP:conf/eurocrypt/DucDF14,DBLP:conf/crypto/BelaidCPRT20,DBLP:conf/crypto/CassiersFOS21} is a well-known setting for security evaluation of cryptographic circuits. The side-channel is noiseless but has a random state: $Y \!=\!(Z,f_Z(K))$ where $f_z$ is a deterministic leakage function depending on state~$z$. Typically $Y$ is either erased or leaks $f(K,T)$ with publicly available plain or cyphertext~$T$.
\begin{theorem}\label{thm:improved-random-probing} In the random probing model,
    \begin{equation}
       1 + \tfrac{M-1}{2} \tfrac{ \exp  ( \tfrac{1-\alpha}{\alpha} H_\alpha(K|Y)  ) - 1}{M^{(1-\alpha)/{\alpha}}-1}
       \lessgtr  G(K|Y) \lessgtr \tfrac{1 + \exp ( H_\alpha(K|Y)) }{2}.
    \end{equation}
where $\lessgtr$ denotes $\geq $ for $\alpha \geq {1}/{2}$ and $\leq$ for $\alpha \leq {1}/{2}$.   
    In the limiting case $\alpha = 1$,
    \begin{equation}
        \frac{1+\exp H(K|Y)}{2} \leq G(K|Y) \leq 1 + \frac{M-1}{2} \frac{H(K|Y)}{\log M}.
    \end{equation} 
\end{theorem}

\begin{proof}
First, consider a fixed function $f = f_z$ with pre-image cardinality $M_y= |f^{-1}( \{ y \} )|$.
Then 
    $H_\alpha(K|Y=y) = \log M_y$ and
    $G(K|Y=y) = \frac{M_y + 1}{2}=\frac{ \exp ( H_\alpha(K|Y=y) ) + 1}{2}= \frac{ K_\alpha(K|Y=y)^{\frac{\alpha}{1-\alpha}} + 1}{2}.$
If $\alpha = 1$, since $x \mapsto \frac{1 + \exp (x) }{2}$ is convex we obtain by Jensen's inequality that 
$G(K|Y) \geq  \frac{\exp \left ( H(K|Y) \right ) + 1}{2}$.
If $\alpha \neq 1$,  $x \mapsto \frac{1 + x^{ \frac{\alpha}{1-\alpha} } }{2}$ is convex 
if $\alpha \geq \frac{1}{2}$ and concave otherwise,
hence $G(K|Y) \lessgtr \frac{ K_\alpha(K|Y)^{\frac{\alpha}{1-\alpha}} + 1}{2}$.
The other bound is obtained by taking the chord of $x \mapsto \frac{1 + x^{ \frac{\alpha}{1-\alpha} } }{2}$.
The region when $f_z$ is obtained at random is obtained by taking the convex hull of the region $(K_\alpha,G)$, which is already convex.
\end{proof}

Fig.~\ref{fig:illus} shows that Theorem~\ref{thm:improved-random-probing} greatly improves Theorem~\ref{th:GrhoHalpha} in the random probing model.
Interestingly, the case $\alpha = \frac{1}{2}$ 
gives equality
$\Delta G(K ; Y)  = \frac{M}{2} ( 1 - \exp ( - I_{\frac{1}{2}}(K;Y) ) ) \approx \frac{M}{2} \frac{I_{\frac{1}{2}}(K;Y)}{\log e}$ where the approximation holds for small leakages.

\begin{figure}[ht!]
    \centering

    \begin{subfigure}{0.4\textwidth}
    \captionsetup{skip=-0.3cm}
    \includegraphics[width=\textwidth,height=.65\textwidth]{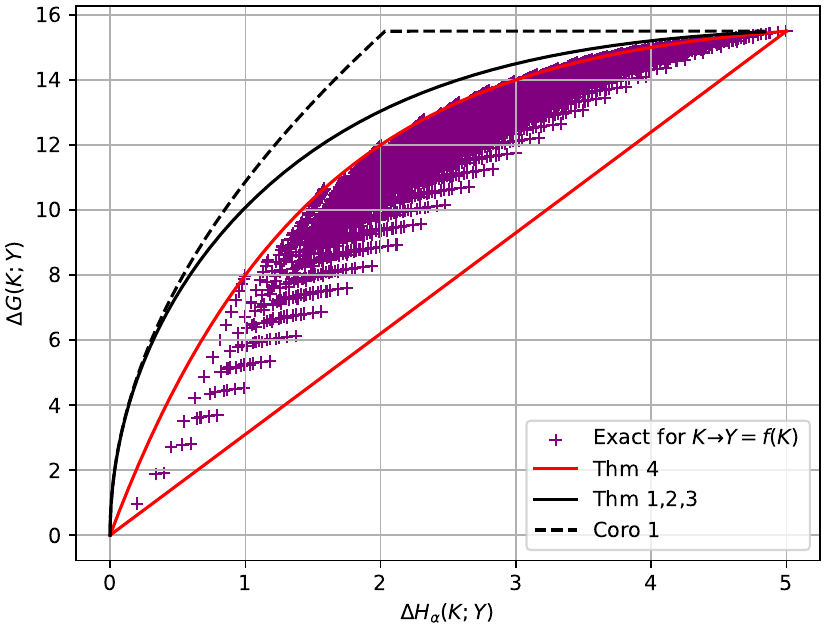}
    \label{fig:illustration_card=8}
    \caption{$\alpha=1$}
    \end{subfigure}

    \begin{subfigure}{0.4\textwidth}
    \captionsetup{skip=-0.3cm}
    \includegraphics[width=\textwidth,height=.65\textwidth]{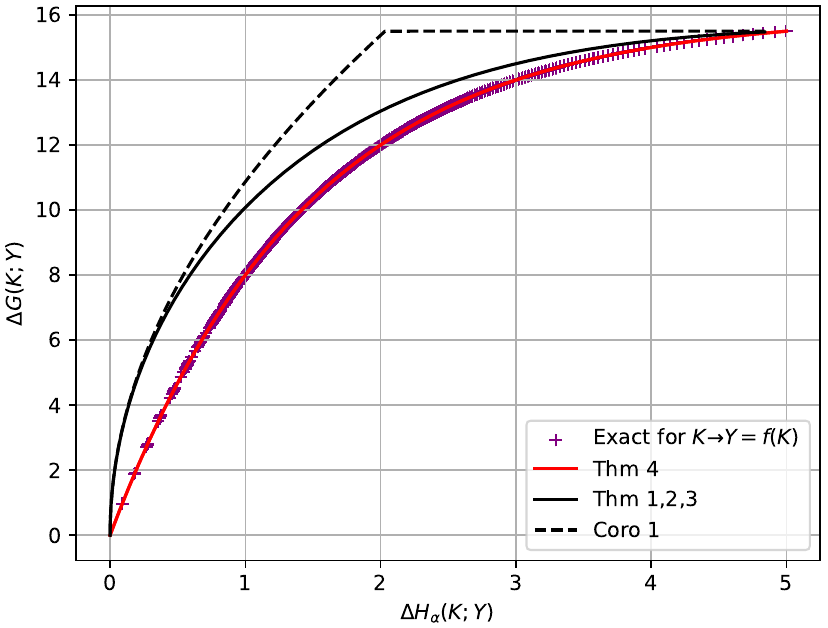}
    \label{fig:illustration_card=32}
    \caption{$\alpha=\frac{1}{2}$}
    \end{subfigure}
    
    \caption{\small Improved bound ($M=32$) from Theorem~\ref{thm:improved-random-probing} (solid, red) compared to Theorem~\ref{th:GrhoHalpha} (solid, black), Corollary~\ref{thm:explicit} (dashed), and scatter plot of exact  values 
    when $Y=f(K)$ for arbitrarily given functions~$f$.
    }
    \label{fig:illus}
\end{figure}


\section{Conclusion}

We derived closed-form optimal regions with explicit bounds on the guessing advantage $\Delta G_\rho$ of a secret random variable in terms of  the $\alpha$-information advantage $\Delta H_\alpha$. 
An important outcome is that it decreases as the square root $\Delta G_\rho\leq c \sqrt{\Delta H_\alpha}$ for small leakages.
%
Simulations in the classical Hamming weight leakage model show that our bounds are much tighter than previous ones, especially for large noise (small leakage).
%
We further sharpened the bounds 
 in the random probing model where the guessing advantage is actually \emph{equal} to 
a function of the Rényi-Arimoto entropy of order~$1/2$.

As possible extensions of this work, it would be valuable to obtain
sharpened bounds for any additive noise leakage model $Y = f(X) + Z$, for secrets with known (nonuniform) prior pmf $p_X$ and generalize to negative values of $\alpha$~\cite{esposito2022sibson}.

%
%
%
\cleardoublepage

\IEEEtriggeratref{18}

\bibliographystyle{IEEEtran} 
\bibliography{main}

\begin{thebibliography}{10}
\providecommand{\url}[1]{#1}
\csname url@samestyle\endcsname
\providecommand{\newblock}{\relax}
\providecommand{\bibinfo}[2]{#2}
\providecommand{\BIBentrySTDinterwordspacing}{\spaceskip=0pt\relax}
\providecommand{\BIBentryALTinterwordstretchfactor}{4}
\providecommand{\BIBentryALTinterwordspacing}{\spaceskip=\fontdimen2\font plus
\BIBentryALTinterwordstretchfactor\fontdimen3\font minus
  \fontdimen4\font\relax}
\providecommand{\BIBforeignlanguage}[2]{{%
\expandafter\ifx\csname l@#1\endcsname\relax
\typeout{** WARNING: IEEEtran.bst: No hyphenation pattern has been}%
\typeout{** loaded for the language `#1'. Using the pattern for}%
\typeout{** the default language instead.}%
\else
\language=\csname l@#1\endcsname
\fi
#2}}
\providecommand{\BIBdecl}{\relax}
\BIBdecl

\bibitem{massey1994guessing}
J.~L. Massey, ``Guessing and entropy,'' in \emph{Proceedings of 1994 IEEE
  International Symposium on Information Theory}, 1994, p. 204.

\bibitem{pliam2000gueswork}
J.~O. Pliam, ``Guesswork and variation distance as measures of cipher
  security,'' in \emph{Selected Areas in Cryptography}, H.~Heys and C.~Adams,
  Eds., 2000, pp. 62--77.

\bibitem{standaert+2009unified}
F.-X. Standaert, T.~G. Malkin, and M.~Yung, ``A unified framework for the
  analysis of side-channel key recovery attacks,'' in \emph{Advances in
  Cryptology - EUROCRYPT 2009}, A.~Joux, Ed., 2009, pp. 443--461.

\bibitem{ChoudaryPopescu17}
M.~O. Choudary and P.~G. Popescu, ``Back to {Massey}: {I}mpressively fast,
  scalable and tight security evaluation tools,'' in \emph{Proc. 19th Workshop
  on Cryptographic Hardware and Embedded Systems (CHES 2017)}, vol. LNCS 10529,
  2017, pp. 367--386.

\bibitem{TanasescuChoudaryRioulPopescu21}
A.~T\u{a}n\u{a}sescu, M.~O. Choudary, O.~Rioul, and P.~G. Popescu, ``Tight and
  scalable side-channel attack evaluations through asymptotically optimal
  {Massey}-like inequalities on guessing entropy,'' \emph{Entropy}, vol.~23,
  no.~11, pp. 1--10, 2021.

\bibitem{arikan1996}
E.~Arikan, ``An inequality on guessing and its application to sequential
  decoding,'' \emph{IEEE Transactions on Information Theory}, vol.~42, no.~1,
  pp. 99--105, 1996.

\bibitem{poussier2016keyenum}
R.~Poussier, F.~Standaert, and V.~Grosso, ``Simple key enumeration (and rank
  estimation) using histograms: An integrated approach,'' in
  \emph{Cryptographic Hardware and Embedded Systems - {CHES} 2016}, ser. LNCS,
  vol. 9813.\hskip 1em plus 0.5em minus 0.4em\relax Springer, 2016, pp. 61--81.

\bibitem{eloi2019info}
{\'E}.~de~Ch{\'e}risey, S.~Guilley, O.~Rioul, and P.~Piantanida, ``An
  information-theoretic model for side-channel attacks in embedded hardware,''
  in \emph{2019 IEEE International Symposium on Information Theory (ISIT)},
  2019, pp. 310--315.

\bibitem{eloi2019ches}
------, ``Best information is most successful: Mutual information and success
  rate in side-channel analysis,'' \emph{IACR Transactions on Cryptographic
  Hardware and Embedded Systems (CHES 2019)}, pp. 49--79, Feb. 2019.

\bibitem{wei2022mask}
W.~Cheng, Y.~Liu, S.~Guilley, and O.~Rioul, ``Attacking masked cryptographic
  implementations: Information-theoretic bounds,'' in \emph{2022 IEEE
  International Symposium on Information Theory (ISIT)}, 2022, pp. 654--659.

\bibitem{issawagnerkamath2020operational}
I.~Issa, A.~B. Wagner, and S.~Kamath, ``An operational approach to information
  leakage,'' \emph{IEEE Transactions on Information Theory}, vol.~66, no.~3,
  pp. 1625--1657, 2019.

\bibitem{julien2023maxleak}
J.~B{\'e}guinot, Y.~Liu, O.~Rioul, W.~Cheng, and S.~Guilley, ``Maximal leakage
  of masked implementations using {Mrs. Gerber}'s lemma for min-entropy,'' in
  \emph{2023 IEEE International Symposium on Information Theory (ISIT)}, 2023,
  pp. 654--659.

\bibitem{arimoto1977information}
S.~Arimoto, ``Information measures and capacity of order $\alpha$ for discrete
  memoryless channels,'' \emph{Topics in information theory (Colloquia
  Mathematica Societatis Janos Bolyai)}, vol.~16, pp. 41--52, 1977.

\bibitem{Sibson69}
R.~Sibson, ``Information radius,'' \emph{Zeitschrift Wahrscheinlichkeitstheorie
  Verwandte Gebiete}, vol.~14, pp. 149--160, Jun. 1969.

\bibitem{Verdu15}
S.~Verd\'u, ``$\alpha$-mutual information,'' in \emph{Proc. Information Theory
  and Applications Workshop}, La Jolla, CA, Feb. 2015, pp. 1--6.

\bibitem{sankar2019alpha}
J.~Liao, L.~Sankar, O.~Kosut, and F.~P. Calmon, ``Robustness of maximal
  $\alpha$-leakage to side information,'' in \emph{2019 IEEE International
  Symposium on Information Theory (ISIT)}, 2019, pp. 642--646.

\bibitem{yi2021conditionalalpha}
Y.~Liu, W.~Cheng, S.~Guilley, and O.~Rioul, ``On conditional alpha-information
  and its application to side-channel analysis,'' in \emph{2021 IEEE
  Information Theory Workshop (ITW)}, 2021, pp. 1--6.

\bibitem{ishai2003private}
Y.~Ishai, A.~Sahai, and D.~Wagner, ``{Private Circuits: Securing Hardware
  against Probing Attacks},'' in \emph{CRYPTO}, ser. Lecture Notes in Computer
  Science, vol. 2729.\hskip 1em plus 0.5em minus 0.4em\relax Springer, August
  17--21 2003, pp. 463--481, {Santa Barbara, California, USA}.

\bibitem{Masure2023nearlytight}
L.~Masure, O.~Rioul, and F.-X. Standaert, ``A nearly tight proof of {Duc et
  al.}'s conjectured security bound for masked implementations,'' in
  \emph{Smart Card Research and Advanced Applications (CARDIS 2022)}, ser.
  LNCS, vol. 13820.\hskip 1em plus 0.5em minus 0.4em\relax Springer Nature,
  2023, pp. 69--81.

\bibitem{Masure2023removing}
J.~B{\'e}guinot, W.~Cheng, S.~Guilley, Y.~Liu, L.~Masure, O.~Rioul, and F.-X.
  Standaert, ``Removing the field size loss from {Duc et al.}'s conjectured
  bound for masked encodings,'' in \emph{Constructive Side-Channel Analysis and
  Secure Design (COSADE 2023)}, ser. LNCS, vol. 13979.\hskip 1em plus 0.5em
  minus 0.4em\relax Springer Nature, 2023, pp. 86--104.

\bibitem{fehrberens2015arimoto}
S.~Fehr and S.~Berens, ``On the conditional r{\'{e}}nyi entropy,'' \emph{IEEE
  Transactions on Information Theory}, vol.~60, no.~11, pp. 6801--6810, 2014.

\bibitem{rioul2022variations}
O.~Rioul, ``Variations on a theme by {Massey},'' \emph{IEEE Transactions on
  Information Theory}, vol.~68, no.~5, pp. 2813--2828, 2022.

\bibitem{mceliece1995inequality}
R.~J. McEliece and Z.~Yu, ``An inequality on entropy,'' in \emph{Proceedings of
  1995 IEEE International Symposium on Information Theory}, 1995, p. 329.

\bibitem{sason2018improved}
I.~Sason and S.~Verd{\'u}, ``Improved bounds on lossless source coding and
  guessing moments via {R\'enyi} measures,'' \emph{IEEE Transactions on
  Information Theory}, vol.~64, no.~6, pp. 4323--4346, 2018.

\bibitem{CoverThomas}
T.~M. Cover and J.~A. Thomas, \emph{Elements of Information Theory}.\hskip 1em
  plus 0.5em minus 0.4em\relax John Wiley \& Sons, 1st Ed. 1990, 2nd Ed. 2006.

\bibitem{rioul2020renyiEPI}
O.~Rioul, ``R{\'{e}}nyi entropy power and normal transport,'' in
  \emph{International Symposium on Information Theory and Its Applications,
  {ISITA} 2020, Kapolei, HI, USA, October 24-27, 2020}, 2020, pp. 1--5.

\bibitem{hirche20renyi}
C.~Hirche, ``R{\'{e}}nyi bounds on information combining,'' in \emph{{IEEE}
  International Symposium on Information Theory, {ISIT} 2020, Los Angeles, CA,
  USA, June 21-26, 2020}, 2020, pp. 2297--2302.

\bibitem{rioul2023interplay}
O.~Rioul, ``The interplay between error, total variation, alpha-entropy and
  guessing: {Fano} and {Pinsker} direct and reverse inequalities,''
  \emph{Entropy}, vol.~25, no.~7, p. 978, 2023.

\bibitem{DBLP:conf/eurocrypt/DucDF14}
A.~Duc, S.~Dziembowski, and S.~Faust, ``{Unifying Leakage Models: From Probing
  Attacks to Noisy Leakage},'' in \emph{{Advances in Cryptology - {EUROCRYPT}
  2014 - 33rd Annual International Conference on the Theory and Applications of
  Cryptographic Techniques, Copenhagen, Denmark, May 11-15, 2014.
  Proceedings}}, ser. LNCS, vol. 8441.\hskip 1em plus 0.5em minus 0.4em\relax
  Springer, 2014, pp. 423--440.

\bibitem{DBLP:conf/crypto/BelaidCPRT20}
S.~Bela{\"{\i}}d, J.~Coron, E.~Prouff, M.~Rivain, and A.~R. Taleb, ``Random
  probing security: Verification, composition, expansion and new
  constructions,'' in \emph{Advances in Cryptology - {CRYPTO} 2020 - 40th
  Annual International Cryptology Conference, {CRYPTO} 2020, Santa Barbara, CA,
  USA, August 17-21, 2020, Proceedings, Part {I}}, ser. LNCS, vol. 12170.\hskip
  1em plus 0.5em minus 0.4em\relax Springer, 2020, pp. 339--368.

\bibitem{DBLP:conf/crypto/CassiersFOS21}
G.~Cassiers, S.~Faust, M.~Orlt, and F.~Standaert, ``Towards tight random
  probing security,'' in \emph{Advances in Cryptology - {CRYPTO} 2021 - 41st
  Annual International Cryptology Conference, {CRYPTO} 2021, Virtual Event,
  August 16-20, 2021, Proceedings, Part {III}}, ser. LNCS, vol. 12827.\hskip
  1em plus 0.5em minus 0.4em\relax Springer, 2021, pp. 185--214.

\bibitem{esposito2022sibson}
A.~R. Esposito, A.~Vandenbroucque, and M.~Gastpar, ``On sibson's
  {\(\alpha\)}-mutual information,'' in \emph{{IEEE} International Symposium on
  Information Theory, {ISIT} 2022, Espoo, Finland, June 26 - July 1, 2022},
  2022, pp. 2904--2909.

\end{thebibliography}

\end{document}